\documentclass[a4paper,12pt]{article}
\usepackage{amsthm}
\usepackage{color}
\usepackage{amssymb,amsmath}
\usepackage[margin=1.2in]{geometry}
\usepackage{graphicx}
\usepackage{amsmath}
\usepackage{listings}
\usepackage{color}

\definecolor{dkgreen}{rgb}{0,0.6,0}
\definecolor{gray}{rgb}{0.5,0.5,0.5}
\definecolor{mauve}{rgb}{0.58,0,0.82}

\title{A note on the Minimum Norm Point algorithm}

\author{
  Igor Stassiy \\
  University of Saarland \\
  Saarbruecken, Germany
}

\begin{document}
\maketitle
\newtheorem{myob}{Observation}
\newtheorem{myde}{Definition}
\newtheorem{mythe}{Theorem}
\newtheorem{myle}{Lemma}
\newtheorem{myco}{Corollary}

\begin{abstract}
We present a provably more efficient implementation of the Minimum Norm Point
Algorithm conceived by Fujishige than the one presented in \cite{FUJI06}. The
algorithm solves the minimization problem for a class of functions known as
submodular. Many important functions, such as minimum cut in the graph, have the
so called submodular property \cite{FUJI82}. It is known that the problem can
also be efficiently solved in strongly polynomial time \cite{IWAT01}, however
known theoretical bounds are far from being practical. We present an improved
implementation of the algorithm, for which unfortunately no worst case bounds
are know, but which performs very well in practice. With the modifications presented, the algorithm performs an order
of magnitude faster for certain submodular functions.
\end{abstract}

\vspace{3em}

\subsection*{Introduction}

Given a base set $S$, a submodular function $F$ is such that, for any $A, B
\subseteq S$ the following holds

\begin{equation}
F(A)+F(B) \geq F(A \cap B)+F(A \cup B)
\end{equation}

It is not hard to show, that a cut in the graph is a submodular function, where
$F(A) = cut(A, V \backslash A)$. The objective is to minimize the cut and which
in turn enables us to find a maximum flow in a graph. It is also known that any
symmetric submodular function, that is for $F(A)=F(S\backslash A) $ for all $A
\subseteq S$, can be seen as a cut function in a certain graph \cite{QUER95}.

\subsection*{Base Polyhedra and Submodular Function Minimization}

Throughout this paper we assume that for a set $E \in 2^{\{1\ldots n\}}$ and a
point $x \in \mathbb{R}^n$ $x(E) = \sum_{e_i \in E} x_{e_i}$ or a sum
of projection on coordinates in $E$. It will also be useful to define a base
polyhedron $B(F)$ with respect to a submodular function $F$:

\begin{myde}
Let $E$ be a finite nonempty set and $F$ be a submodular function $F : 2^E
\mapsto \mathbb{R}$. Suppose that $F(\emptyset) = 0$, then we can define the
base polyhedron:

\begin{eqnarray*}
P(F) & = & \{x | x \in \mathbb{R}^E, \forall X \in 2^E : x(X) \leq F(X)\} \\
B(F) & = & \{x | x \in P(F), x(E) = F(E)\}
\end{eqnarray*}
\end{myde}

\subsection*{Minimum Norm Point Algorithm}
\newcommand{\argmin}{\operatornamewithlimits{argmin}}
Suppose we are given a finite set $P$ of points $p_i \in \mathbb{R}^n$. The
problem is to find the minimum norm point $x^{*}$ in the convex hull of points $p_i$
i.e. $\argmin \| x\|_2, x \in CH(p_1, p2_, \ldots, p_n)$.
The following theorem establishes the relationship between minimum norm point in
the convex hull and the minimization of a submodular function:

\begin{mythe}
Let $x^{*}$ be the minimum norm point in the base polyhedron $B(F)$ as defined
above. Define 
\begin{eqnarray*}
A_{+} & = & \{e | e \in E, x^{*}(e) \leq 0\} \\
A_{-} & = & \{e | e \in E, x^{*}(e) < 0\}
\end{eqnarray*}

Then $A_{+}$ is the unique maximal minimizer of $F$ and $A_{-}$ is the unique
minimal minimizer of $F$.
\end{mythe}

Equipped with this theorem we can find the minimum norm point in the base
polyhedron $B(F)$ and find the minimum of a submodular function $F$.

Here follows the description of the minimum norm point algorithm: 
Throughout the runtime, the algorithm maintains a {\emph simplex} of points
$\in 2^P$ and a current minimum norm point $\hat{x}$. With each update of the
simplex, the norm of $\hat{x}$ decreases.
\\
\\
{\emph Input:} A finite set of points $P = \{p_1, p_2, \ldots p_k\}, p_i \in
\mathbb{R}^n$ \\
{\emph Output:} The minimum norm point $x^{*}$ in the convex hull $\hat{P}$ of the
points $\{p_1, \ldots , p_k\}$
\begin{enumerate}
  \item Choose any point $p \in P$ and put $S = {p}$ and $\hat{x} = p$.
  \item Find a point $\hat{p} \in P$ that minimizes the linear function $\langle
  \hat{x}, p\rangle = \sum_i\hat{x}_i p_{i}$. If $\langle
  \hat{x}, p\rangle = \langle \hat{x}, \hat{x}\rangle$, return $x^{*} = \hat{x}$.
  Else go to step $3$.
  \item Find the minimum norm point $y$ in the affine hull of points in $S$. If
  $y$ lies in the relative interior of the convex hull of $S$, then put
  $\hat{x} = y$ and go to step $2$.
  \item Let $z$ be the point that is the nearest to $y$ among the intersection
  of the convex hull of points in $S$ and the line $[y, \hat{x}]$ between $y$
  and $\hat{x}$. Additionally, let $S' \subset S$ be the unique proper subset 
  of $S$ such that $z$ lies in the relative interior of the convex hull of $S'$.
  Put $S = S'$ and $\hat{x} = z$. Go to step $3$.
\end{enumerate}

The cycle formed by the steps $2 \leftrightarrow 3$ is called a \emph{major}
cycle and the one by steps $3 \leftrightarrow 4$ a \emph{minor} cycle. In
major/minor cycles, the \emph{simplex} size increases/decreases
correspondingly. In major cycle the \emph{simplex} increases by $1$ and in the
minor decreases by at least $1$. 

\begin{myde}
A \emph{simplex} is called a \emph{corral} if the minimum norm point lies in the
relative interior of the convex hull of the points of the \emph{simplex}.
\end{myde}

\begin{myle}
Every \emph{corral}  uniquely  determines the current minimum norm point. 
\end{myle}

\begin{myle}
After at most $n-1$ iterations in the minor cycle, the current \emph{simplex}
becomes a \emph{corral}. 
\end{myle}
\begin{proof}
As there will be left at most $1$ point in the \emph{simplex} S.
\end{proof}
\begin{myle}
After each iteration of step $2$ the norm of the $\hat{x}$ is decreasing.
\end{myle}

\begin{mythe}
The described minimum norm point algorithm terminates in a finite number of
steps. It is currently open to decide if the algorithm runs in polynomial time.
\end{mythe}

\subsection*{Implementation}

Step $2$ of the algorithm requires a linear optimization, which can be done by
computing $\langle \hat{x}, p\rangle = \sum_i\hat{x}_i p_{i}$ for all the points
in $P$, however the number of points can be exponential. 

In the case the set $P$ is given implicitly, such as a number of extreme
points of a polytope $Q$. 

Luckily, for base polyhedra associated with submodular functions this problem
can be solved greedily as was shown by Edmonds: \\
\\
{\emph Input:} $w \in R^E$, submodular function $F$ \\
{\emph Output:} An optimal $x^{*} \in B(F)$ that minimizes $\sum_{e \in E}w(e)x(e)$
\begin{enumerate}
  \item Find an ordering of $e_1, e_2, \ldots, e_n$ s.t. 
    \begin{equation}
    w(e_1) \leq w(e_2) \leq  \ldots w(e_n)
    \end{equation}
  \item Compute $x^{*}$ as follows:
    \begin{equation}
    x^{*}(e_i) = F(\{e_1, e_2, \ldots e_i\})-F(\{e_1, e_2, \ldots , e_{i-1}\}), \:
    (i = 1, 2, \ldots n)
    \end{equation}    
\end{enumerate}

\begin{myle}
The resulting $x^{*}$ lies in the base polyhedron $B(F)$ and minimizes $\sum_{e \in
E}w(e)x(e)$
\end{myle}

Step $3$ requires solving the following optimization problem:

\begin{eqnarray*}
\min & \| x \| & \\
x & = & \sum_{1 \leq i \leq n} \alpha_i p_i \\
\sum_{1 \leq i \leq n} \alpha_i & = & 1, p_i \in P, \alpha_i \in \mathbb{R}
\end{eqnarray*}

Equivalently, the problem can we rewritten as

\begin{eqnarray*}
\min & \| x \| & \\
x & = & p_1 + \sum_{1 \leq i \leq n} \alpha_i (p_i-p_1) \: \Leftrightarrow
\: p_1+\sum_{2 \leq i \leq n} \alpha_i (p_i-p_1) \\
\sum_{1 \leq i \leq n} \alpha_i & = & 1, p_i \in P, \alpha_i \in \mathbb{R}
\: \Leftrightarrow \: p_i \in P, \alpha_i \in \mathbb{R}
\end{eqnarray*}

Consider the subspace of vectors $p_i-p_1$. Let $p_1 =
p_1^{\parallel}+p_1^{\perp}$, such that $\langle p_1^{\perp}, p_i-p_1 \rangle =
0$ for all $2 \leq i \leq n$.
\begin{myle}
There exists a unique such decomposition $p_1 =
p_1^{\parallel}+p_1^{\perp}$
\end{myle}

Then it follows that 

\begin{myle}
Denote $S_P$ the subspace of vectors $p_i-p_1, 2 \leq i \leq n$. Also, let
$v \in \mathbb{R}^n$ belong to $S_P$, s.t. $v = p_1^{\parallel}+\sum_{2 \leq i
\leq n} \alpha_i (p_i-p_1)$. For a minimum norm point $x$,
\begin{equation}
\min \| x \| = \|p_1^{\perp}+p_1^{\parallel}+\sum_{2 \leq i \leq n} \alpha_i (p_i-p_1) \| = 
\| p_1^{\perp} \|+\| v \| \geq \| p_1^{\perp} \|
\end{equation}
\end{myle}

Clearly, the inequality is tight and holds for $v = \vec{0}$, hence the
optimization problem is minimized for $x = p_1^{\perp} = p_1-p_1^{\parallel}$.
Note that the choice of $p_1$ was completely arbitrary, and any vector $p_i$
could be chosen.

Finding $p^{\parallel}$ is known as projection onto the subspace and can be
found as follows:

\begin{myle}
Let $M = \begin{bmatrix}
  p2-p1 & p3-p1 & \ldots & p_n-p1 \\
 \end{bmatrix} \in R^{n(k-1)}$, then projection of $p_1$ onto the subspace is 
 
 \begin{equation}
 p_1^{\parallel} = M(M^{T}M)^{-1}M^{T}p_1
 \end{equation}
\end{myle}

It is however inefficient to find the projection in this way, as we would need
to compute the inverse of the \emph{hat} matrix $(M^{T}M)$. Instead, we can
solve the following system of equations:

\begin{eqnarray}
(M^{T}M) y = M^{T}p_1 \\
p_1^{\parallel} = M y 
\end{eqnarray}

The system of equations above is usually solved using the Gaussian elimination
process of the matrix on the left hand side and back substitution subsequently. 
For general matrices, straightforward Gaussian elimination requires $O(n^3)$
operations or more precisely $\approx 2/3 n^3$ operations. Matrix vector multiplication takes $n k$
operations where the $n, k$ are the matrix dimensions.

The step $4$ of the algorithm, we can determine the point $z$ as follows:
\begin{myle}
Let 
\begin{equation}
\hat{x} = \sum_i \lambda_i p_i, \: y = \sum_i \mu_i p_i
\end{equation} 
then $z$ can be determined such that 
\begin{equation}
z = (1-\beta)\hat{x}+\beta y, \:
(1-\beta)\lambda_i+\beta \mu_i \geq 0, \forall i
\end{equation} and $\beta$ is large as possible.
\end{myle} 

\subsection*{Improvement idea}

The improvement is based on the following idea known as Sherman-Morrison-Woodbury matrix
inverse update:

\begin{myle}
Let $M = \begin{bmatrix}
  A & U \\
  V & D
\end{bmatrix}$ then, $M^{-1} = \begin{bmatrix}
  A^{-1}+A^{-1}UC^{-1}V A^{-1} & -A^{-1} U C^{-1} \\
  -C^{-1} V A^{-1} & C^{-1}
\end{bmatrix}$ where $C = D-V A^{-1} U$
\end{myle}

Now, notice that during the runtime of our algorithm, in the steps $2
\leftrightarrow 3$ we only add $1$ column to the matrix $S$ and in the steps $3
\leftrightarrow 4$ we delete at least $1$ column of the matrix. With the lemma 
above, we could update the inverse of the matrix $M^{T}M$ and solve the
system of equations $(6), (7)$ more efficiently.

\subsection*{Precise formulation of updates}

As was demonstrated above, it is possible to update the matrix inverse using the
blockwise approach. The matrix that we are dealing with is of the form $M^T M$.

Without loss of generality, suppose that a column $v$ is appended to the matrix
$M$ as the last column, i.e. $M' = \begin{bmatrix} M & v \\ \end{bmatrix}$. Let
us call such an update a \emph{rank-up} update. Then, $M'^{T} M' =
\begin{bmatrix} M^{T}M & M^{T} v
\\
v^T M & v^T v \end{bmatrix}$

Then, simply substituting $A = M^{T}M$, $V = v^{T}M$, $U = M^{T}v$ and $C = v^T
v$ we can apply the Sherman-Morrison-Woodbury matrix
inverse update formula.

Now, suppose, the last $k$ columns of the matrix $M$ are removed and let the new
matrix be $N$. We call such an update a \emph{rank-down} update. Then $M =
\begin{bmatrix} N & K \\
\end{bmatrix}$ and $M^{T} M = \begin{bmatrix} N^{T}N & N^{T}K \\ K^{T}N & K^{T}K \end{bmatrix}$.

Now, notice that, in the Sherman-Morrison-Woodbury update, the lower right,
upper right and lower left blocks of the matrix multiplied in the following way:

\begin{equation}
(A^{-1} U C^{-1})(C^{-1})^{-1}(-C^{-1} V A^{-1}) = A^{-1}UC^{-1}V A^{-1}
\end{equation}

which is exactly the term of the upper left block of the matrix 

\begin{equation}
M^{-1} = \begin{bmatrix}
  A^{-1}+A^{-1}UC^{-1}V A^{-1} & -A^{-1} U C^{-1} \\
  -C^{-1} V A^{-1} & C^{-1}
\end{bmatrix}
\end{equation}

Hence, knowing the inverse of the matrix $M^{T} M
= \begin{bmatrix} N^{T}N & N^{T}K \\ K^{T}N & K^{T}K \end{bmatrix}$, we can find
out the inverse of the matrix $N^{T}N$:

Let 
\begin{equation}
(M^{T} M)^{-1} = \begin{bmatrix} P & Q \\ Q^{T} & R \end{bmatrix}
\end{equation}

Then, the inverse of $N^{T}N = P-Q^T R^{-1} Q$. Note that, we are again faced
with the problem of taking the inverse of a matrix $R$ and matrix
multiplications. Note that a rank-down update by $k$ columns can be realized by
a series of $k$ rank-down updates which remove only a single column. 

\begin{myle}
The running time of a single rank-up or rank-down operations that add or remove
a single column is $O(n^2)$.
\end{myle}
\begin{proof}
Let us firstly consider the rank-up update. Matrices $V = v^{T}M$ and $U
= M^{T}v$ can be computed in time $O(n^2)$ as computing them corresponds to
matrix-vector multiplications. The matrix $D = v^{T}v$ is computable in time
$O(n)$ and the matrix $C$ in $O(n^2)$.
With similar reasoning, the product $A^{-1}UC^{-1}V
A^{-1} = (A^{-1} U C^{-1})(C^{-1})^{-1}(-C^{-1} V A^{-1})$ can be computed in
$O(n^2)$, given that the multiplications are realized as suggested by the
placement of brackets. Hence, overall the time to do a rank-up update is
$O(n^2)$. We can apply the very same techniques to verify that a rank-down
update by $1$ column  is implementable in $O(n^2)$ time.
\end{proof}

The lemma that follows implies that the efficient updates
presented above make it possible to carry all the inverse updates an order of magnitude faster than in
the original algorithm.

\begin{myle}
The amortized cost of rank-up and rank-down in arbitrary sequence of operations
is $O(n^2)$. And hence, the total running time of a sequence of length $t$ of
rank-up and rank-down updates takes time $O(tn^2)$.
\end{myle}
\begin{proof}
Recall that the running time of the algorithm is dominated by the total number
of times the steps $2$ and $3$ are called, multiplied by the time a respective
step takes. In a \emph{major} cycle we add a column to the matrix $M$ and during
a \emph{minor} cycle a number of columns are removed. In the step $2$ we
would need to solve an optimization problem, which will be discussed later.
In the $3$-rd step the algorithm needs to solve a system of a kind $y = Ax$ for
a given $y$. When an update of $A$ is readily available, this can be done in
$O(n^2)$ time, however without it we would need $O(n^3)$ time.

Let the number of removed columns during the $i$-th \emph{minor} cycle be $k_i$
and the number of times the \emph{major} cycle is called $m_i$, then the total
running time of the algorithm is $O(\sum_i m_in^2+\sum_i k_i n^2)$, as a rank
update of any kind of a single column takes $O(n^2)$ time.

Although the number of operations during a \emph{minor} cycle can be as bad as
$O(n^3)$ when $k_i = O(n)$, the amortized time for every rank-up and rank-down update
(independent of the number of columns that are deleted) can be shown to be $O(n^2)$.
This can be shown using the accounting method of amortized analysis. Let every rank-up
update bring a $O(n^2)$ to the system and another $O(n^2)$ to pay for its own
update. Then when a rank-down happens, every column has funds to pay for its
rank-down update. Hence, the amortized time for every update is $O(n^2)$ and any
sequence of such updates is computable in time $O(tn^2)$ where $t$ is the length
of the sequence.

Note that, the original algorithm also needs to solve a
system of a kind $y = Ax$ for a given $y$, however without the update of $A$
readily available, which needs $O(n^3)$ time.
And hence when $t$ such steps need to be performed, the total time is $O(tn^3)$, while the
improved version requires only $O(tn^2)$ time.
\end{proof}

\section*{Acknowledgements}
I would like to thank Prof. Dr. Matthias Hein for proposing this question and
for fruitful discussions.

\end{document}